\documentclass[sigconf]{aamas}

\usepackage{balance} %

\setcopyright{ifaamas}
\acmConference[AAMAS '22]{Proc.\@ of the 21st International Conference
on Autonomous Agents and Multiagent Systems (AAMAS 2022)}{May 9--13, 2022}
{Online}{P.~Faliszewski, V.~Mascardi, C.~Pelachaud,
M.E.~Taylor (eds.)}
\copyrightyear{2022}
\acmYear{2022}
\acmDOI{}
\acmPrice{}
\acmISBN{}

\usepackage [english]{babel}
\usepackage [autostyle, english = american]{csquotes}
\MakeOuterQuote{"}

\usepackage{float}

\usepackage{todonotes}

\usepackage{amsmath}
\usepackage{amsthm}
\usepackage{subcaption}
\usepackage{wrapfig,graphicx}
\usepackage{color}
\usepackage{comment}
\usepackage{mathtools}
\usepackage{float}
\usepackage{amsfonts}
\usepackage{siunitx} %
\usepackage{dblfloatfix} %

\definecolor{gray}{rgb}{0.35,0.35,0.35}
\definecolor{blue}{rgb}{0,0,1}
\definecolor{red}{rgb}{1,0,0}
\definecolor{orange}{rgb}{0.75, 0.4, 0}
\definecolor{green}{rgb}{0.0, 0.5, 0.0}

\newtheorem*{proposition*}{proposition}

\newcommand{\ignore}[1]{}

 \def\A{\mathcal{A}}

\newcommand{\Cpp}{C\raise.08ex\hbox{\tt ++}\xspace}

\newtheorem{lemma}{Lemma}
\newtheorem{theorem}{Theorem}
\newtheorem{corollary}{Corollary}

\theoremstyle{definition}

\theoremstyle{plain}

\def\0{\bm{0}}

\makeatletter
\def\thmhead@plain#1#2#3{%
  \thmname{#1}\thmnumber{\@ifnotempty{#1}{ }\@upn{#2}}%
  \thmnote{ {\the\thm@notefont#3}}}
\let\thmhead\thmhead@plain
\makeatother

\usepackage{multirow}
\usepackage{pdfrender}
\newcommand*{\boldcheckmark}{%
  \textpdfrender{
    TextRenderingMode=FillStroke,
    LineWidth=.5pt, %
  }{\checkmark}%
}

\newcommand{\poslit}{\ensuremath{R^+}}
\newcommand{\neglit}{\ensuremath{R^-}}

\newcommand{\src}[1] {\ensuremath{s(#1)}}
\newcommand{\trg}[1] {\ensuremath{t(#1)}}
\newcommand{\pth}[1] {\ensuremath{\pi_{#1}}}
\newcommand{\pthdef}[2] {\ensuremath{\pi_{#1}(#2)}}
\newcommand{\gadg}{\ensuremath{\Xi}}

\newcommand{\mycomment}[1]{%
}%

\acmSubmissionID{783}

\title[Refined Hardness of Distance-Optimal Multi-Agent Path Finding]{Refined Hardness of \\Distance-Optimal Multi-Agent Path Finding}

\author{Tzvika Geft}
\affiliation{
  \institution{Tel Aviv University}
  \city{Tel Aviv}
  \country{Israel}}
\email{zvigreg@mail.tau.ac.il}

\author{Dan Halperin}
\affiliation{
  \institution{Tel Aviv University}
  \city{Tel Aviv}
  \country{Israel}}
\email{danha@post.tau.ac.il}

\begin{abstract}
We study the computational complexity of multi-agent path finding (MAPF).
Given a graph $G$ and a set of agents, each having a start and target vertex, the goal is to find collision-free paths minimizing the total distance traveled.
To better understand the source of difficulty of the problem, we aim to study the simplest and least constrained graph class for which it remains hard.
To this end, we restrict $G$ to be a 2D grid, which is a ubiquitous abstraction, as it conveniently allows for modeling well-structured environments (e.g., warehouses).
Previous hardness results considered highly constrained 2D grids having only one vertex unoccupied by an agent, while
the most restricted hardness result that allowed multiple empty vertices was for (non-grid) planar graphs.
We therefore refine previous results by simultaneously considering both 2D grids and multiple empty vertices.
We show that even in this case distance-optimal MAPF remains NP-hard, which settles an open problem posed by Banfi et al.~\cite{DBLP:journals/ral/BanfiBA17}.
We present a reduction directly from 3-SAT using simple gadgets, making our proof arguably more informative than previous work in terms of potential progress towards positive results.
Furthermore, our reduction is the first linear one for the case where $G$ is planar, appearing nearly four decades after the first related result.
This allows us to go a step further and exploit the Exponential Time Hypothesis (ETH) to obtain an exponential lower bound for the running time of the problem.
Finally, as a stepping stone towards our main results, we prove the NP-hardness of the monotone case, in which agents move one by one with no intermediate stops.
\end{abstract}


\newcommand{\BibTeX}{\rm B\kern-.05em{\sc i\kern-.025em b}\kern-.08em\TeX}

\begin{document}

\pagestyle{fancy}
\fancyhead{}

\maketitle

\begin{table*}[t]
\caption{Comparison of previous hardness proofs for distance-optimal MAPF. Exact definitions for parallel and sequential motion are given in Section~\ref{sec:terminology}. We define the term \textit{linear reduction} in Section~\ref{sec:ETH}.}
\label{tab:comp}
\centering

\begin{tabular}{lcccccc}
\midrule
 Paper & Planar & Grid subgraph & $>1$ Empty vertex & Parallel & Sequential & Linear reduction  \\
 \midrule
Goldreich~\cite{Goldreich84}                            &   &   &   &   & \checkmark & \checkmark \\
Ratner and   Warmuth~\cite{ratner1990n2}                & \checkmark & \checkmark &   &   & \checkmark &   \\
Demaine and   Rudoy~\cite{DBLP:journals/tcs/DemaineR18} & \checkmark & \checkmark &   &   & \checkmark &   \\
Yu and   LaValle~\cite{YuGraphs}                        &   &   &   & \checkmark &   &   \\
Yu~\cite{YuPlanar}                                      & \checkmark &   & \checkmark & \checkmark & \checkmark &   \\

\midrule
\textbf{this work}                                              &\boldcheckmark&\boldcheckmark&\boldcheckmark&\boldcheckmark&\boldcheckmark&\boldcheckmark\\
     
\end{tabular}
\end{table*}

\section{Introduction}
We are witnessing the rapid development and adaptation of autonomous multi-robot systems in a wide variety of application domains.
Such systems are deployed in warehouse logistics~\cite{app:warehouses, app:amazonkiva}, rail traffic scheduling~\cite{app:trains}, autonomous aircraft towing~\cite{app:towing}, and many more~\cite{DBLP:conf/socs/FelnerSSBGSSWS17}.
Successfully deploying multi-robot systems requires algorithms for efficiently planning collision-free motions, which is an ever-growing research field.

We study the motion planning problem of multiple agents operating on a grid.
Grid graphs are commonly employed in multi-robot motion planning as a simple means of environment discretization.
There is also particular interest in structured environments, such as warehouses, which are often already grid-like by design (see, e.g.,~\cite{app:warehouses, ma2017multi}).
In this work we are concerned with intractability of optimizing the sum of distance traveled by agents on grids.
Minimizing distance has been widely studied in motion planning~\cite{DBLP:conf/focs/CannyR87, DBLP:conf/rss/SoloveyYZH15, DBLP:conf/cccg/KirkpatrickL16, DBLP:journals/arobots/ShomeSDHB20}.
The objective has also been studied in MAPF grid domains~\cite{yu2018constant, DBLP:journals/arobots/Yu20, DBLP:journals/ral/WangR20}, including being one of the objectives of the SoCG~2021 Challenge~\cite{fekete2021computing}.

The goal of this work is to make the hardness analysis for distance-optimal MAPF more comprehensible and applicable.
Our driving force closely ties in with future research directions presented in a recent AAMAS blue-sky paper~\cite{DBLP:conf/atal/SalzmanS20}.
Along with others~\cite{DBLP:conf/socs/FelnerSSBGSSWS17, gordon2021revisiting}, it advocates the need for better understanding the hardness of MAPF.

\textbf{Previous work.}
The hardness of \textit{time}-optimal MAPF has been well-studied over the previous decade~\cite{DBLP:conf/aaai/Surynek10, YuGraphs, DBLP:conf/aaai/MaTSKK16, YuPlanar}, including results for grid graphs~\cite{DBLP:journals/ral/BanfiBA17, DBLP:journals/siamcomp/DemaineFKMS19}.
Previous work on \textit{distance}-optimal MAPF dates back at least as far as the 1980s, but has not treated modern MAPF formulations as extensively.
We now outline results on the intractability of distance-optimal MAPF and compare them based on a few parameters: the type of graph, the number of empty vertices, i.e., vertices not containing an agent, and whether parallel motions are allowed.

The roots of the problem can be traced back to the classic 15-puzzle~\cite{johnson1879notes15puzz, wilson1974graph}, which can viewed as moving 15 agents on a 16-vertex grid graph.
In 1984 Goldreich~\cite{Goldreich84, DBLP:books/sp/goldreich2011/Goldreich11} presented the first NP-hardness result for the generalized puzzle, which consisted of a general graph with one empty vertex.
Ratner and Warmuth~\cite{ratner1990n2} refined his result to hold for the 2D-grid, known as the $(n^2-1)$-puzzle, using a more complicated reduction from a special SAT variant.
A simpler hardness reduction for the same problem from Rectilinear Steiner Tree~\cite{garey1977rectilinear} was recently shown by Demaine and Rudoy~\cite{DBLP:journals/tcs/DemaineR18}.
The results so far all consider a very constrained case where only a single vertex is unoccupied by an agent and only one agent can move at a time.

Meanwhile, the problem has evolved and modern MAPF formulations started allowing multiple agents to move together, including simultaneous rotation of agents along fully occupied cycles~\cite{DBLP:conf/wafr/YuR14}.
Accounting for this, Yu and LaValle~\cite{YuGraphs} introduced these motions into the intractability analysis.
They showed that the problem remains NP-hard for parallel motions, including rotations, for general graphs where all vertices are occupied.
Yu~\cite{YuPlanar} later refined these results to hold for planar graphs with more than one empty vertex (under the same parallel motions).
In Table~\ref{tab:comp} we give a concise comparison of all the previous proofs.

The problem formulation by Yu~\cite{YuPlanar} can be considered as the one closest to the prevalent formulation of MAPF on grids.
Nevertheless, Yu's intractability result does not imply hardness for grids since planar graphs are more general, and his construction cannot be readily adapted to grids.
Indeed, the grid case has been subsequently posed as an open problem by Banfi et al.~\cite{DBLP:journals/ral/BanfiBA17}, who showed the NP-hardness of time-optimal MAPF on 2D grids.

\textbf{Contribution.}
We settle the latter question by showing that distance-optimal MAPF remains hard on the 2D grid even with multiple empty vertices.
By considering the simplest graph environment with more movement freedom, our result essentially establishes hardness for the easiest variant of the problem thus far.
Since we show hardness for a more specific class of graphs than the previous result by Yu~\cite{YuPlanar}, our proof applies to their case as well.
Our hardness proof is via a direct reduction from 3-SAT using simple gadgets.
Its simplicity is highlighted by the fact that it is a linear reduction, which stands in contrast to all previous proofs (except for Goldreich's~\cite{Goldreich84}, which uses a non-planar graph with only one empty vertex).

Although our result is negative in nature, we argue that its relative simplicity compared to related proofs has important practical implications that are in-line with current research goals.
To this end, we discuss the main benefits of simple hardness proofs.

First, they make it easier to highlight the parameters that make the problem hard.
By closely identifying such parameters, one can evaluate different algorithms with respect to these parameters and improve algorithm selection~\cite{DBLP:conf/aips/KaduriBS20}. %
Indeed, for the related time objectives, it has been noted that there is no algorithm that dominates all the others~\cite{DBLP:conf/socs/FelnerSSBGSSWS17}, motivating such comparisons.
Furthermore, identifying such parameters can pave the way towards positive results using parameterized complexity~\cite{cygan2015parameterized}.
As a modern approach for tackling NP-hard problems, parameterized complexity been recently raised as a potential research direction for MAPF~\cite{DBLP:conf/atal/SalzmanS20}.
Under this approach, we aim for exact yet efficient algorithms that are exponential only in the size of a
fixed parameter while being polynomial in the size of the input.
Lastly, capturing the hardness of the problem in an easy to grasp way can provide algorithm designers an intuitive basis for better solutions.
We discuss observations in this spirit based on our hardness construction in the conclusion.

Of additional practical significance is our establishment of a concrete lower bound for distance-optimal MAPF.
Such bounds provide a crucial indication on whether running times of algorithms can be improved.
While NP-hardness results provide evidence that computational problems are unlikely to be solvable in polynomial time, the underlying complexity assumption, namely P != NP, does not give any concrete lower
time bounds.
Indeed, many NP-hard problems differ widely in hardness in practice.
Therefore, a stronger assumption is needed for more meaningful results.

A nowadays common assumption for this purpose is the Exponential Time Hypothesis (ETH), introduced by Impagliazzo and Paturi~\cite{DBLP:journals/jcss/ImpagliazzoP01}.
Roughly speaking, it conjectures that 3-SAT cannot be solved in subexponential time $2^{o(N)}$, where $N$ is the number of variables.
The ETH has far reaching consequences (see, e.g., the survey~\cite{DBLP:journals/eatcs/LokshtanovMS11})
leading to increased adoption in robotics and artificial intelligence~\cite{DBLP:journals/jair/BackstromJ17, DBLP:conf/aaai/EibenGKY18, DBLP:journals/amai/AghighiBJS16}.

Obtaining an exponential lower bound using the ETH requires more fine-grained hardness reductions that do not blow-up the size of the resulting instance (which is roughly the number of agents in our case).
Since our reduction has a linear number of agents, unlike previous ones for the planar case, we are able to obtain the first exponential lower bound for distance-optimal MAPF.

\textbf{Organization.}
As a stepping stone, we show the hardness of a more restricted problem version, called monotone MAPF, in which agents move one by one to their targets and each agent is only allowed to move once.
The monotone version of the problem arises in the context of object \textit{rearrangement}, in which a robot moves set of objects one by one from a given configuration to another~\cite{wang2021uniform}.

In Section~\ref{sec:terminology} we introduce our terminology and problem definition.
In Section~\ref{sec:ETH} we give background on the ETH and how we can use it to obtain lower bounds.
In Section~\ref{sec:monotone} we show the hardness of monotone distance-optimal MAPF, which is adapted to the general (non-monotone) case in Section~\ref{sec:non-monotone}.

\section{Terminology} \label{sec:terminology}
We now define distance-optimal MAPF.
We are given an undirected graph $G(V, E)$ and a set $R$ of $n$ agents.
Each agent $r \in R$ has a start vertex $\src{r} \in V$ and goal vertex $\trg{r} \in V$.
We define a {\em trajectory (timed path)} for an agent $r$ as a sequence
$\pth{r}: \mathbb {N} \to V$ where $\mathbb N$ is the set of non-negative integers representing time steps.
A feasible $\pth{r}$ must be a sequence of vertices that connects $\src{r}$ and $\trg{r}$: 
\begin{itemize}
    \item  $\pthdef{r}{0} = \src{r}$
    \item  $\exists T_i \in \mathbb N$, s.t. $\forall \tau \ge T_i, \pthdef{r}{\tau} = \trg{r}$
    \item $\forall \tau > 0$, $ \pthdef{r}{ \tau-1} =  \pthdef{r}{\tau}$ or $(\pthdef{r}{\tau-1},  \pthdef{r}{\tau}) \in E$
\end{itemize}
We call the set of trajectories for all agents $\{\pth{r}\}_{r\in R}$ a \emph{motion plan}.
We call the motion plan {\em collision-free} if and only if the agents do not simultaneously occupy the same vertex or edge.
That is, $\forall r, r' \in R$ s.t. $r \neq r'$, $\pth{r}, \pth{r'}$ must satisfy the following: 
\begin{itemize}
    \item $\forall \tau \ge 0$, $\pthdef{r}{\tau} \neq \pthdef{r'}{\tau}$
    \item $\forall \tau > 0$, $(\pthdef{r}{\tau-1},  \pthdef{r}{\tau}) \ne (\pthdef{r'}{\tau},  \pthdef{r'}{\tau-1})$. 
\end{itemize}

\textbf{Monotone motion plan.}
The \textit{active interval} of an agent $r$ in a motion plan, denoted by $I_r$, is the interval from the first time $r$ leaves $\src{r}$ to the last time $r$ reaches $\trg{r}$, i.e.,
\[I_r \coloneqq [\min_{\tau \in \mathbb N} \pthdef{r}{\tau+1} \neq \src{r}, \max_{\tau \in \mathbb N} \pthdef{r}{\tau-1} \neq \trg{r}]\]
If the active intervals $\{I_{r}\}_{r\in R}$ of a motion plan are pairwise disjoint then we call it a \textit{monotone} motion plan, i.e., the agents move one by one.

\begin{sloppypar}
\textbf{Paths and distance cost.}
We define the \emph{path} of an agent $r$ in a motion plan, denoted by $P(r)$, to be its path in $G$ in the regular graph theoretic sense, i.e., \pth{r} with consecutive appearances of the same vertex $v$ replaced by a single occurrence of $v$.
The \emph{length} of a path is the number of edges in the path.
The \emph{distance cost} of a motion plan is the sum of the lengths of the paths of the agents in $R$, i.e., the total distance traveled by the agents.
For an instance ${M \coloneqq (G,R,
\{\src{r},\trg{r}\}_{r\in R})}$, 
we denote by $d^*(M)$ the cost of a motion plan where each agent takes the shortest possible path, i.e., this is the optimistic lower bound for the distance cost.
Formally,
\mbox{$d^*(M)= \sum_{r\in R} d(\src{r},\trg{r})$} where $d(u,v)$ is the distance, namely, the length of a shortest path, between $u$ and $v$ in $G$.
\end{sloppypar}

We can now define the two decision problems for which we prove NP-hardness:

\medskip
\noindent\textbf{Distance-Optimal MAPF}:
Given $G, R, \{\src{r},\trg{r}\}_{r\in R}$ 
as defined above and an integer $k \in \mathbb N$, is there a motion plan $\{\pth{r}\}_{r\in R}$ that is collision-free and has a distance cost of at most $k$?
\medskip

\medskip
\noindent\textbf{Monotone Distance-Optimal MAPF}:
Same as above, except that the motion plan needs to be monotone.
\medskip

\textbf{Remark: Parallel, sequential, and monotone plans.} %
The above (non-monotone) formulation allows the strongest notion of parallel synchronized motions.
For example, notice that an agent can move into a vertex that is just being left by another agent, i.e., agents can move like a train.
Specifically, this also allows agents to synchronously rotate along a fully occupied cycle.
In \emph{sequential} MAPF (classically known as pebble motion on graphs), the motion plan can only have one agent moving at each time step.
Note that a monotone motion plan is also sequential, but that opposite is not true, i.e., a plan can be sequential but not monotone.

\section{Lower Bounds Using the Exponential Time Hypothesis} \label{sec:ETH}
When designing or improving an algorithm for a problem, a natural question is, "What is the fastest possible algorithm for the problem?"
A common way of addressing the problem is using the theory of NP-hardness, which uses the assumption that P != NP.
Under this assumption, it is widely believed that NP-hard problems such as 3-SAT cannot be solvable in polynomial time.
Unfortunately, the assumption is too weak to allow us to conclude any concrete lower bounds.
Therefore, a stronger assumption called the Exponential Time Hypothesis (ETH) was introduced by Impagliazzo and Paturi~\cite{DBLP:journals/jcss/ImpagliazzoP01}.
In essence, it relies on research barriers as evidence for the nonexistence of a sub-exponential algorithm for 3-SAT.
With this stronger assumption, the ETH has enabled to delineate NP-complete problems based on concrete time bounds, which is more fine grained than the usual classification into complexity classes.

We now state the version of the hypothesis that we will use.
The original ETH states the lower bound in terms of the number of variables $N$ of the 3-SAT formula.
However, the output instances of hardness reductions usually depend on the size of the formula, i.e., the number of literals, which could be as large as $O(N^3)$.
Therefore, through the use of the Sparsification Lemma by Impagliazzo et al.~\cite{DBLP:journals/jcss/ImpagliazzoPZ01-sparse} it was shown that the original ETH also holds with respect to the number of clauses, which we state as follows:

\medskip
\noindent\textbf{Exponential Time Hypothesis} \cite{DBLP:journals/jcss/ImpagliazzoP01, DBLP:journals/jcss/ImpagliazzoPZ01-sparse}.
There is no algorithm solving every instance of 3-SAT with $N$ variables and $M$ clauses in time $2^{o(N+M)}$.
\medskip

Using this hypothesis we are able to obtain concrete lower bounds as follows.
First, for convenience, we note that for a 3-SAT formula $\phi$ we have $|\phi| = O(N+M)$, where $|\phi|$ is the size of the formula.
Consider a \textit{linear} reduction from 3-SAT to some problem $A$, i.e., a polynomial-time algorithm that takes a 3-SAT formula $\phi$ and outputs an equivalent instance $x$, whose size, $|x|$, is bounded by $O(|\phi|)$.
Then, if $A$ had an algorithm with a running time of $2^{o(|x|)}$, we could use it, after applying the reduction, to solve 3-SAT in time $2^{o(|\phi|)}=2^{o(N+M)}$.
Therefore, the existence of a linear reduction from 3-SAT to $A$ implies the nonexistence of a $2^{o(|x|)}$ algorithm for $A$ under the ETH.
We will present linear reductions in order to make the same claim for distance-optimal MAPF.

In general, a reduction from 3-SAT to $A$ outputting an instance of size $O(g(|\phi|))$ would exclude an $2^{o(f(|x|))}$-time algorithm for $A$, where $f$ is the inverse of $g$.
Therefore, reductions aiming to obtain lower bounds using ETH should keep the \textit{blow up} of the instance, represented by the function $g$, as close to linear as possible.

\begin{figure*}[!t]
\centering
\includegraphics[width=1\textwidth]{figures/sum_dist_monotone.pdf}
\caption{The instance $M$ for the formula $\phi = (\overline{x} \lor \overline{y} \lor z) \land (x \lor \overline{y} \lor  z) \land (x \lor y \lor \overline{z})$.
Obstacles appear in gray.
The leftmost variable gadget's entrance and exit are marked by a cross and a dot, respectively.
The start and target positions are the filled  and unfilled colored squares, respectively. Positive and negative literal agents (and their target positions) are green and red, respectively.
Literal agents are labeled with unique indices in order to distinguish between appearances of the same literal.
Clause agents and their target positions are cyan.
}
\label{fig:monotone}
\end{figure*}

\begin{figure*}[!t]
\centering
\includegraphics[width=1\textwidth]{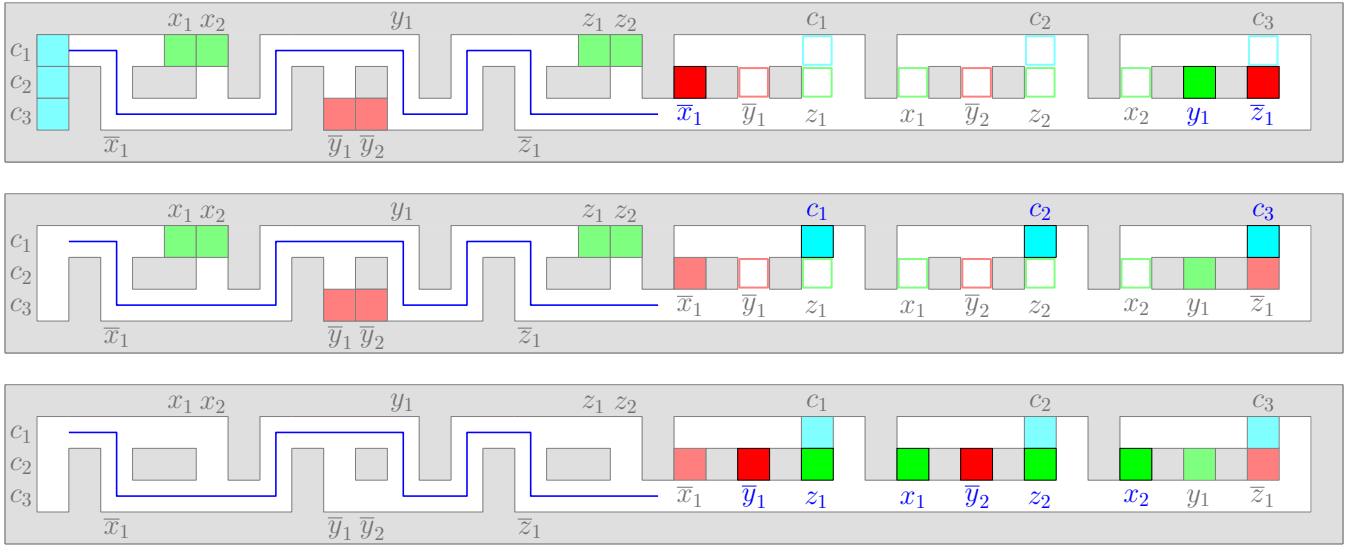}
\caption{
Three stages of an optimal monotone motion plan for the instance in Figure~\ref{fig:monotone}.
The plan corresponds to the assignment $x=T, y=F, z=T$, for which
$P$ is shown in blue.
We use blue text labels to highlight target positions of agents that have moved in each stage.
In the first stage (top) agents in \neglit{} move, in the following order: $\overline{z}_1, y_1, \overline{x}_1$ (namely $\overline{z}_1$ moves first, then $y_1$ and finally $\overline{z}_1$). 
In the next stage (middle) the clause agents move: $c_1, c_2, c_3$.
In the final stage (bottom), the agents in \poslit{} move: $z_2, z_1, \overline{y}_2, \overline{y}_1, x_2, x_1$.}
\label{fig:monotone_plan}
\end{figure*}

\section{Monotone Distance-Optimal MAPF} \label{sec:monotone}
In this section we prove that Monotone Distance-Optimal MAPF is NP-hard for grid graphs and present an exponential lower bound under the ETH.
We remark that the feasibility problem, i.e., simply deciding whether a monotone (collision-free) motion plan exists, was shown to be NP-hard~\cite{geft2021complexity}.
Therefore, we refine the problem definition by assuming here that some monotone motion plan exists and focus only on the hardness of optimization.

Before presenting the reduction, we provide an intuition into the hardness of the problem.
Recall that once an agent $r$ reaches its target, it may no longer move.
As a result, when $r$ is at its target, agents that move after it might be forced to detour around $r$ and take a longer path.
Hence, an algorithm for the problem needs to carefully decide which agent to move next, so that shortest paths that will be needed (for agents moving later) are not blocked.
The decision is further complicated by the fact that before an agent $r$ can move, some agents may need to move first to clear a path for $r$.

We now present a reduction from 3-SAT, the problem of
deciding satisfiability of a formula in conjunctive normal form with three literals in each clause.
The reduction is an adaptation of a simplified version of the hardness proof for monotone MAPF~\cite{geft2021complexity}.
Given a 3-SAT formula $\phi$, we construct a corresponding monotone MAPF instance $M$ that has a motion plan with the lowest attainable distance cost, $d^*(M)$, if and only if $\phi$ is satisfiable.

An example of the construction is shown in Figure~\ref{fig:monotone}, which should be followed throughout the description.
We present the figure as a planar workspace composed of unit grid cells, which are dual to the vertices on a grid graph.
We therefore occasionally use the term \textit{cell} to refer to a vertex in the graph.

The construction is a grid $G$ with three rows.
It contains two types of agents: \emph{literal agents}, each corresponding to an appearance of a literal in $\phi$, and \emph{clause agents}, each corresponding to a clause of $\phi$.
The clause agents are initially located in a rectangular "room" at the very left of the construction (see cyan squares in Figure~\ref{fig:monotone}).
To their right is a series of \textit{variable gadgets}, which the clause agents have to \textit{traverse} though, i.e., enter and exit in the general rightward direction.
Each gadget has an \textit{entrance} on the left and an \textit{exit} on the right, which are marked by a cross and a dot, respectively, in Figure~\ref{fig:monotone}.
We use obstacle cells to make entrances, exits, and other passages only one row/column wide.

The literal agents' start positions are located in the variable gadgets.
Each variable gadget initially contains literal agents of a single variable and has two optimal-length paths for traversing it.
The \textit{top path} (i.e., the one going right along the grid's top row, then down) initially contains positive literal agents, and the \textit{bottom path} path initially contains negative literal agents.
Note that there are some empty cells in variable gadgets in Figure~\ref{fig:monotone} that are not necessary for the current construction, but will later play a role in Section~\ref{sec:non-monotone} (and are kept for commonality between figures).

The literal agents' target positions are located in \emph{clause gadgets}.
Each clause gadget's middle row contains the target positions of the literals in the clause that the gadget represents.
The gadget also contains a target of one clause agent on its top row.
We include an empty column at the right of each clause gadget to ensure accessibility of the latter target.
Specifically, the empty columns allow each clause agent $c$ to reach its target position even if all the targets in the respective clause gadget are occupied (note that in this case $c$'s path will be longer than the shortest possible path).

The gadgets are arranged from left to right so that first we have variable gadgets and then clause gadgets. The order of gadgets of the same type and start/target positions within a gadget are arbitrary.

\begin{figure*}[!t]
\centering
\includegraphics[width=1\textwidth]{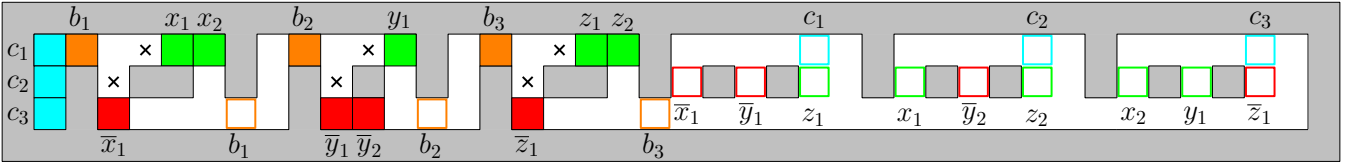}
\caption{
The instance $M'$ modified from $M$ in Figure~\ref{fig:monotone}.
Blocker agents (and their target positions) are shown in orange.}
\label{fig:non-monotone}
\end{figure*}

It is easy to verify that there exists a monotone motion for plan for $M$. Simply repeat the following for each variable gadget \gadg{}, going from the rightmost to the leftmost gadget: first move agents that are on \gadg{}'s top path in right to left order and then do the same for \gadg{}'s bottom path.
At this point all literal agents are at their targets.
Therefore, from now on there always exists a clause agent that may be moved to its target, until all have reaches their targets.

The following theorem proves the correctness of the construction.
\begin{theorem} \label{thm:monotone}
$M$ has a monotone motion plan with a distance cost of $d^*(M)$ if and only if $\phi$ is satisfiable.
\end{theorem}

\begin{proof}
  Assume that $\phi$ has a satisfying assignment $\mathcal{A}$.
  Let \poslit{} (resp. \neglit{}) be the set of agents corresponding to literals that evaluate to true (resp. false) according to $\mathcal{A}$.
  Let $P$ be the shortest path from the entrance of leftmost variable gadget to the exit of the rightmost variable gadget that passes through all the start positions of \neglit{}; see Figure~\ref{fig:monotone_plan}.
  Observe that for each variable gadget, \neglit{} contains agents that are all either on the gadget's top path or bottom path. This means that $P$ exists and that it is $x$-monotone.
  
  We specify a monotone motion plan in which all agents move along $P$ while traversing variable gadgets.
  The motion plan has three stages, which are illustrated in Figure~\ref{fig:monotone_plan}. In each stage a group of agents move, starting with \neglit{}, then the clause agents, and finally \poslit{}.
  First, agents in \neglit{} move in right to left order along $P$, which guarantees no collisions between literal agents.
  Observe that each agent in \neglit{} can achieve the shortest path to its target, initially guided by $P$.
  Next, the clause agents move in the natural order that allows each of them to leave their initial room using the shortest path with no collisions.
  We have the following properties at this point: $P$ contains only empty cells and each clause gadget's middle row must also contain an unoccupied target of an agent in \poslit{}.
  The latter holds because $\mathcal{A}$ satisfies $\phi$ and the agents of \poslit{} have not yet moved.
  Therefore, each clause agent can also take an optimal path.
  Finally, \poslit{} can move optimally, guided by $P$, similarly to \neglit{}.

  For the other direction, we assume that there is a monotone motion plan for $M$ with a distance cost of $d^*(M)$ and show that $\phi$ has a satisfying assignment.
  Let \poslit{} denote the agents that move after the last clause agent moves.
  For any variable $\alpha \in \phi$, \poslit{} cannot contain literal agents corresponding to both $\alpha$ and $\overline{\alpha}$, since then clause agents would not be able to reach their target positions.
  Therefore, we can define an assignment $\mathcal{A}$ in which the literals corresponding to \poslit{} evaluate to true.
  (If a variable does not have literals in \poslit{}, then it can be assigned an arbitrary value.)
  
  Let $C$ be a clause in $\phi$ and let $c$ be the corresponding clause agent, i.e., $c$ has to go to $C$'s clause gadget.
  There must be a target vertex $v$ in the middle row of $C$'s clause gadget that is unoccupied when $c$ moves.
  Such a vertex $v$ must exist in order for $c$ to have the shortest possible path to its target $\trg{c}$ during its turn to move.
  Therefore, the literal agent $r$, with $\trg{r}=v$ must be in \poslit{} by definition, i.e., it must move after $r$.
  Hence, the literal corresponding to $r$ evaluates to true by $\mathcal{A}$, which means that $C$ is satisfied and we are done.
\end{proof}

It is easy to verify that the number of agents in $M$ as well as the size of the resulting graph is linear in $|\phi|$. Therefore, we conclude the following.

\begin{corollary}
Monotone Distance-Optimal MAPF is NP-hard and cannot be solved in sub-exponential time $2^{o(n)}$ or $2^{o(|V|)}$ unless ETH fails, even for a grid graph $G=(V, E)$ with 3 rows, where $n$ is the number of agents.
\end{corollary}

\begin{figure*}[!t]
\centering
\includegraphics[width=1\textwidth]{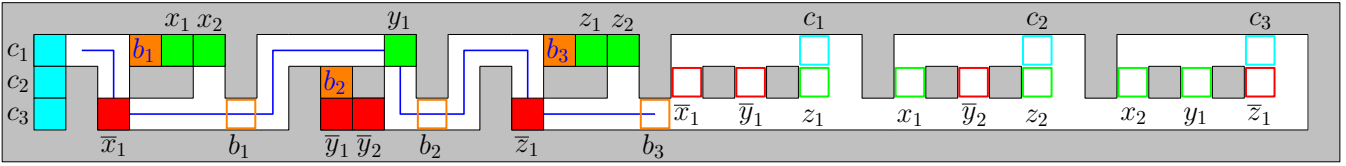}
\caption{
The path $P$ (blue) corresponding to the assignment $x=T, y=F, z=T$ along with a snapshot after the first stage of the motion plan for $M'$.}
\label{fig:non-monotone-stage1}
\end{figure*}

\section{General Distance-Optimal MAPF} \label{sec:non-monotone}
Our previous construction no longer holds once non-monotone motions are allowed.
Since agents are not constrained by time, they can make intermediate stops, which adds a lot of possibilities to the motion plan.
The main challenge is that literal agents can ``cheat'' by making intermediate stops in variable gadgets along their way.
For example, in Figure~\ref{fig:monotone} $y_1$ could position itself in the cell to the left of $z_1$,
thereby creating a path through $y$'s variable gadget that does not enforce an assignment to $y$.\footnote{Eliminating free cells in the two paths in $z$'s variable gadget does not solve the problem.
Observe that it is possible for both positive and negative literal agents to leave the gadget before any clause agent moves. If this happens, it would create space for the undesirable intermediate stops described.}
Therefore, we introduce \textit{blockers}, which are new agents that prevent undesirable intermediate stops, and prove that general distance-optimal MAPF is NP-hard on grid graphs.

As before, for a 3-SAT formula $\phi$, we construct a distance-optimal MAPF instance $M'$ that has a motion plan with a distance cost of $d^*(M')$ if and only if $\phi$ is satisfiable.
An example of the new instance $M'$ is illustrated in Figure~\ref{fig:non-monotone}.
In general, $M'$ is the same as $M$ from Section~\ref{sec:monotone} except for the following change:
Each variable gadget now has a blocker agent that starts at the gadget's entrance and has to go to the gadget's exit.
This ensures that all the agents passing through the gadget must use the same path within the gadget, thereby keeping the incoming and outgoing order of the traversing agents the same.
This property prevents clause agents from "cheating" and bypassing literal agents, thereby mimicking the monotone case.
The following lemma formally states the functionality of the blockers:

\begin{lemma} \label{lem:must-use-P}
Let \gadg{} be a variable gadget in $M'$.
Then, in any motion plan for $M'$ with a cost of $d^*(M')$, all the agents that traverse \gadg{} must take the same path through \gadg{}.
\end{lemma}

\begin{proof}
    Let $b$ denote the blocker agent that is initially at \gadg{}'s entrance.
    Since $P(b)$ is an optimal path, it can be either the top path or the bottom path in \gadg{}, which we denote by $P_1$ and $P_2$, respectively; see Figure~\ref{fig:non-monotone-lem}.
    Similarly, any agent $r$ that traverses \gadg{}, must have either $P_1$ or $P_2$ be a subpath of its path, $P(r)$.
    However, we cannot have $P(b)$ be a subpath of $P(r)$, as that would necessarily lead to a collision between $b$ and $r$.
    That is, $r$ must somehow bypass $b$ to exit \gadg{}, which is not possible if they take the same path in \gadg{}. %
    This leaves $r$ with exactly one path that it can take through \gadg{}, namely, the one not taken by $b$.
    Since this applies to any agent $r$, we have all the agents traversing \gadg{} take the same path through \gadg{}, as required.
\end{proof}

\begin{figure}[H]
\centering
\includegraphics[width=0.18\textwidth]{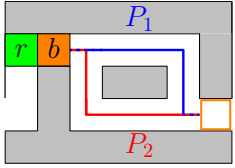}
\caption{
A variable gadget at some time point during a motion plan.
The gadget's respective blocker agent $b$ is at its start position and its target position is on the right. The agent $r$ needs to traverse the gadget. The two possible shortest paths for $b$ are shown as $P_1$ (blue) and $P_2$ (red). Note that these paths overlap near \src{b} and \trg{b}. Here we highlight that $b$ and $r$ cannot use the same path in the gadget.}
\label{fig:non-monotone-lem}
\end{figure}

We now prove the correctness of the modified construction.

\begin{theorem} \label{thm:non-monotone}
$M'$ has a motion plan with a distance cost of $d^*(M')$ if and only if $\phi$ is satisfiable.
\end{theorem}
\begin{proof}
We adjust the motion plan defined for $M$ in the proof of Theorem~\ref{thm:monotone} to accommodate the blocker agents.
Let $\mathcal{A}$ be a satisfying assignment and let $P$ be as defined before.
The new plan has two additional stages.
First, blocker agents move to intermediate stops that are not on $P$ (in arbitrary order).
Figure~\ref{fig:non-monotone} indicates the two possible intermediate stops for each blocker agent using crosses and Figure~\ref{fig:non-monotone-stage1} shows an example of the blockers' positions after the first stage.
Next, we perform the same motion plan as for $M$.
Then, in the final stage, blockers move to their targets (again in arbitrary order, since each assignment only contains the blocker at this stage).
It is easy to verify that the new plan is optimal: The intermediate stops always allow blocker agents to not block $P$, while also allowing them to eventually reach their targets using optimal paths.
As for the rest of the agents, each agent is able to take the same path as in Theorem~\ref{thm:monotone}.

For the other direction, let us assume that there is a motion plan for $M'$ with a cost of $d^*(M')$.
By Lemma~\ref{lem:must-use-P}, all the clause agents follow the same path in each variable gadget.
Therefore, let $P$ denote the path that all the clause agents follow between the entrance of the leftmost variable gadget to the exit of the rightmost variable gadget.
We define a satisfying assignment $\A$ using $P$ as follows:
A variable $\alpha \in \phi$ is assigned to be true (resp. false) if $P$ passes though start positions of negative (resp. positive) literal agents in $\alpha$'s variable gadget.
In other words, literals corresponding to literal agents that are initially located on $P$ are assigned to be false (see Figure~\ref{fig:non-monotone-stage1} for an example of the correspondence between $P$ and $\A$).
As before, let \poslit{} (resp. \neglit{}) be the set of agents corresponding to literals that evaluate to true (resp. false) according to $\mathcal{A}$.

Let $C$ be a clause in $\phi$ and let $c$ be the corresponding clause agent, i.e., $c$ has to go to $C$'s clause gadget.
We show that $C$ is satisfied by $\A$.
It suffices to show that $c$'s path, $P(c)$, contains a target of an agent in \poslit{} in $C$'s clause gadget.
Let us assume for a contradiction that this does not hold.
Then, since $P(c)$ is $c$'s individually optimal path, it must still pass through a target in $C$'s clause gadget.
This target must be $\trg{r}$ of some $r \in \neglit{}$.
We will now claim that $P(r)$ must be a subpath of $P(c)$.
Intuitively, this means that $c$ cannot bypass $r$, and will ultimately be blocked by $r$ once $r$ reaches $\trg{r}$, thus yielding the contradiction.

Let \gadg{} be the variable gadget on which \src{r} lies.
By definition, \src{r} lies on $P$, so $r$ must follow $P$ to exit \gadg{} using the shortest path.
By Lemma~\ref{lem:must-use-P}, $r$ must continue following $P$, the subpath shared by all clause agents, in all variable gadgets it traverses (after leaving \gadg{}).
The remainder of $P(r)$ must also be a subpath of $P(c)$ since both paths are optimal and hence simply go right until reaching the cell below $\trg{r}$.
Therefore, $P(r)$ as a whole is a subpath of $P(c)$. This is a contradiction since then $r$ must reach $\trg{r}$ before $c$ does, which would block $c$.
In conclusion, we showed that $C$ is satisfied, which holds for any clause, and so we are done.
\end{proof}

Observe that all our arguments hold regardless of whether parallel motion is allowed or not.
For the case of synchronous rotations along cycles, by definition for such a rotation to occur, there has to be an agent moving left.
As none of the individually optimal paths for agents ever require moving left, rotations cannot occur for a plan with cost $d^*(M')$.
It is easy to verify that the number of agents in $M'$ as well as the size of the resulting graph remains linear in $|\phi|$.
Therefore, we conclude the following:

\begin{corollary}
Distance-Optimal MAPF is NP-hard and cannot be solved in sub-exponential time $2^{o(n)}$ or $2^{o(|V|)}$ unless ETH fails, even for a grid graph $G=(V, E)$ with 3 rows, where $n$ is the number of agents. This holds for both parallel and sequential motions.
\end{corollary}

\section{Conclusion}
We have shown that distance-optimal MAPF is NP-hard on grid graphs with more than one empty vertex, settling the open problem by Banfi et al.~\cite{DBLP:journals/ral/BanfiBA17}.
Before discussing possible implications of our proof towards positive results, we advocate for more focus on achieving refined hardness results. %
Specifically, we believe that when proving hardness, one should strive towards the following two goals:
The considered problem setting should be the most restricted, i.e., simplest, one that is still useful, while at the same time the hardness reduction should be kept simple as well, ideally with low blow-up.
We note that these two goals can sometimes collide (e.g., compare the elegant proof by Goldreich~\cite{Goldreich84} for general graphs versus the one for the grid~\cite{ratner1990n2}, which was simplified only close to 30 years later~\cite{DBLP:journals/tcs/DemaineR18}).
Nevertheless, we believe that we have homed in on both goals in this paper, thereby improving the structural understanding of MAPF.

\subsection{Implications of the hardness result}
The previous hardness proof for distance-optimal MAPF on planar graphs by Yu~\cite{YuPlanar} uses agents that need to move in opposite directions in order to emulate an assignment.
As a result, Yu concluded that the hardness of the problem appears to arise from contention that occurs when two or more groups of agents want to move in opposite directions through the same set of narrow paths. %
From a practical standpoint, Yu suggests that environments with many robots would benefit from a design that minimizes path sharing among the robots.

Since in our construction all the agents move in the same general direction, we show that hardness remains even without opposite direction movement.
In fact, we remark that our construction can be modified so that the problem is NP-hard even if agents can only move down and right. This requires two main modifications: The variable gadgets need to be arranged in a staircase-like manner, in which the exit of one gadget is on the same grid row as the entrance of its neighboring gadget. This modification eliminates the need for agents to go up between variable gadgets. The second modification is vertically mirroring the clause gadgets such that the clause agents' targets are on the bottom row each gadget.

Given that opposite direction movement does not play a role in our case, we provide another perspective for the source of difficulty of the problem.
For the purpose of this discussion, when we say that a target vertex $v$ is \emph{fulfilled}, we mean that the agent $r$ with $\trg{r}=v$ has reached $v$.
Our construction's challenging aspect is the need of agents to negotiate through paths with many start and target vertices of other agents.
This results in two conflicting goals:
On the one hand each agent needs to pass target positions along its path before they become fulfilled (assuming that once they become fulfilled, the agent will have to take a longer path).
This suggests that algorithms for distance-optimal MAPF can benefit from "prioritizing" agents that have targets along their optimal path that are close to becoming fulfilled.
At the same time, each agent should aim not to force other agents to move in a manner that fulfills targets that other agents still need to pass through.

\subsection{Future work}
Our discussion on implications of the hardness results calls for the investigation of more parameters that affect the hardness of the problem.
As we noted, agents in our construction have to pass through a large number of start and target positions of other agents.
A natural question is whether the problem remains hard even if each agent has an optimal path that passes through a constant number of start and target positions.
Another significant feature of our construction is that the agents' paths must largely overlap.
Therefore, the case where the paths can overlap less, which requires a different layout than the "long and narrow" grid that we used, seems worthy of more study.
Overall, we believe that more subtle underlying parameters need to be considered.
Previous positive results that employ parameterized complexity in discrete motion planning problems~\cite{DBLP:journals/jair/GuptaSZ20, DBLP:conf/soda/AgarwalAGH21} provide some encouragement.

While our refined analysis has resulted in a concrete lower bound, we are not aware of any algorithm that matches (or nearly matches) it, i.e., has a running time of $O(2^n)$ or $O(2^{|V|})$.
This brings to light a gap between lower and upper bounds, which we believe calls for additional refined analysis on both sides.
On the upper bound side such work was recently done by~\citet{gordon2021revisiting} for (time-optimal) Conflict-Based Search~\cite{DBLP:journals/ai/SharonSFS15}, which tightened the running time of the algorithm.
Their improved bound is exponential in a few parameters, therefore it could be beneficial to simultaneously analyze multiple parameters on the lower bound side.
We also remark here that existing hardness results for time-optimal MAPF on planar graphs~\cite{YuPlanar} and 2D grid graphs~\cite{DBLP:journals/ral/BanfiBA17, DBLP:journals/siamcomp/DemaineFKMS19} use reductions that are not linear. 
Hopefully, tightening both lower and upper bounds will uncover areas for algorithmic improvements.

\begin{acks}
This work has been supported in part by the Israel Science
Foundation (grant no.~1736/19),
by NSF/US-Israel-BSF (grant no.~2019754),
by the Israel Ministry of Science and Technology (grant no.~103129),
by the Blavatnik Computer Science Research Fund, and by
the Yandex Machine Learning Initiative for Machine Learning
at Tel Aviv University.
\end{acks}

\bibliographystyle{ACM-Reference-Format} 
\balance
\bibliography{references}

\end{document}